\documentclass[11pt,a4paper]{scrartcl}

\usepackage{multirow}
\usepackage{hyperref}
\usepackage{amsmath,amssymb,amsthm}
\usepackage{mathtools}

\usepackage[utf8]{inputenc}
\usepackage[T1]{fontenc}

\usepackage{booktabs}
\usepackage{tikz} 
\usepackage{mathrsfs}

\usepackage{thm-restate}

\newtheorem{theorem}{Theorem}
\newtheorem{definition}[theorem]{Definition}
\newtheorem{lemma}[theorem]{Lemma}

\newtheorem{observation}[theorem]{Observation}
\newtheorem*{technique-no-number}{Technique}
\newtheorem{corollary}[theorem]{Corollary}
\newtheorem{example}[theorem]{Example}

\DeclareMathOperator{\cost}{\mathsf{cost}}
\DeclareMathOperator{\poly}{\mathsf{poly}}
\DeclareMathOperator{\wcost}{\mathsf{wcost}}
\DeclareMathOperator{\opt}{\mathsf{opt}}
\DeclareMathOperator{\dist}{dist}

\newcommand{\kk}{\texorpdfstring{$k$}{k}}
\newcommand{\R}{\mathbb{R}}
\newcommand{\N}{\mathbb{N}}

\newcommand{\E}{\mathbb{E}}
\newcommand{\V}{\mathbb{V}}

\usepackage[bottom=1in,top=1in,left=1in,right=1in]{geometry}
\RedeclareSectionCommand[indent=0pt]{subparagraph}

\title{Coresets for \\constrained \kk-median and \kk-means clustering\\ in low dimensional Euclidean space}

\author{Melanie Schmidt\thanks{University of Cologne, Germany, \href{mailto:mschmidt@cs.uni-koeln.de}{mschmidt@cs.uni-koeln.de}}\ \ and Julian Wargalla\thanks{University of Cologne, Germany, \href{mailto:wargalla@cs.uni-koeln.de}{wargalla@cs.uni-koeln.de}}}

\allowdisplaybreaks

\begin{document}

\maketitle

\begin{abstract}
    We study (Euclidean) $k$-median and $k$-means with constraints in the streaming model.
    There have been recent efforts to design unified algorithms to solve constrained $k$-means problems without using knowledge of the specific constraint at hand aside from mild assumptions like the polynomial computability of feasibility under the constraint (compute if a clustering satisfies the constraint) or the presence of an efficient assignment oracle (given a set of centers, produce an optimal assignment of points to the centers which satisfies the constraint). These algorithms have a running time exponential in $k$, but can be applied to a wide range of constraints.
    We demonstrate that a recently proposed technique for solving a specific constrained $k$-means problem, namely fair $k$-means clustering, actually implies streaming algorithms for all these constraints. The technique uses the computation of so-called coresets. A coreset is a small summary that approximately preserves the cost of all feasible solutions. A special type of coreset for which constructions were proposed at the beginning of the coreset era actually satisfy a very strong type of guarantee which implies their applicability to general constrained $k$-means problems. The drawback is that these constructions introduce an exponential dependence on the dimension $d$.
    However, for constant $d$, one immediately gets a streaming algorithm which computes a coreset for any constrained $k$-means problem (under above mild assumptions) and which can be combined with approximation algorithms to get a streaming approximation algorithm. We note that our paper builds heavily on previous work, yet we believe that it is worthwhile to know that streaming coresets for constrained $k$-means (and Euclidean $k$-median) clustering can be obtained fairly easily independently of the specific constraint.
\end{abstract}

\section{Introduction}

Let $(X, \dist)$ be a metric space. We consider a general (sum-based) $k$-clustering problem that asks, given a finite subset $P \subseteq X$ and a positive integer $m$, for a set of $k$ centers $C \subseteq X$ that minimizes the cost function
\begin{align}
    \label{formulationone}
    \cost(P,C) = \sum_{p \in P} \min_{c \in C} \dist(p, c)^m.
\end{align}
As two examples of different flavor that are both individually prominent, we consider the \emph{$k$-median problem} where $m=1$ and $X=P$ and the \emph{$k$-means problem} where $m=2$, $X=\R^d$ and $\dist$ is the Euclidean distance.
Both problems are APX-hard~\cite{G00,JMS02,ABS10,LSW17}, and the best known approximation algorithms achieve an approximation ratio of $2.675+\varepsilon$~\cite{BPRST17} for $k$-median and 6.357~\cite{ANSW17} for $k$-means.
If $k$ is a constant, then both problems allow for a $(1+\varepsilon)$-approximation~\cite{HPM04,KSS10}.
Since the purpose of this paper is mainly streaming algorithms, we consider the special case of \emph{Euclidean} $k$-median, which makes talking about streaming a bit more concise. This means that $P\subset \mathbb{R}^d$ and $\dist$ is the Euclidean metric. The best known approximation algorithm for Euclidean $k$-median achieves a factor of $2.633+\varepsilon$~\cite{ANSW17}.

These results hold for the standard formulations of the problems. When \emph{constraints} are added to the picture, $k$-clustering problems are much less understood. There is a multitude of possible constraints. If we think of the clusters as groups of people that attend the same supermarket or are in the same district, then we want a clustering where no cluster has too many points, i.e., we want to enforce capacities. These may be global (no more than a certain amount of people per district) or individual (this school can accept this many students). If we think of the clustering as a means of anonymization where in the end we output centers as representatives of their clusters and publish these, then we want no cluster to be too small. If we think of a scenario where our data comes from sensors and has measuerment errors, we may want that we get a clustering that has the flexibility to disregard some points as measurement errors. These are only some examples. We know constant-factor approximations for some clustering problems with constraints, for example~\cite{AS12-lowerbounded,CKLV17,KLS18,S10}, yet it is for example a major open problem whether capacitated $k$-median admits a constant-factor approximation.

In this paper we are interested in general results for large classes of constraints. In general, we can model a constraint via the question of which clusterings are allowed under that constraint. For this view, it is convenient to think of a solution of a $k$-clustering problem as a partitioning of $P$ into $k$ clusters.
It is well-known that for $k$-means, the optimal center for a cluster $C_i$ is then the centroid $\mu(C_i)=\frac{1}{|C_i|} \sum_{x \in C_i} x$. This is why sometimes the $k$-means problem is also formulated as follows: Given $P$, compute a partitioning into clusters $C_1,\ldots,C_k$ such that
\begin{align}
    \label{formulationtwo}
    \sum_{i=1}^k \dist^2(C_i,\mu(C_i))
\end{align}
is minimized, where we use the abbreviation $\dist^2(C,z)=\sum_{x \in C} ||x-z||^2$. For any solution $c_1,\ldots,c_k$ to \eqref{formulationone}, we can get a solution of the same or lower cost with respect to \eqref{formulationtwo} if we partition the points by assigning every point to its closest center in $\{c_1,\ldots,c_k\}$, and for any solution $C_1,\ldots,C_k$ for \eqref{formulationtwo}, the centers $\mu(C_1),\ldots,\mu(C_k)$ are a solution for \eqref{formulationone} of the same cost. The same can be done for $k$-median, where if we represent a solution as a partitioning, we implicitly assume that for each cluster a center is chosen optimally (since $X=P$, an optimal center can be found efficiently).

In this partitioning-based view, constraints can be modeled via the partitionings that are allowed, i.e.,  the solution space is restricted to only these partitionings.
Three relatively simple, but often considered variants are clustering with lower bounds, capacities and outliers.

\begin{example}
    In a clustering problem with lower bounds we are given values $\ell_1, \ldots, \ell_k$ and are only allowed to produces clusterings $C_1, \ldots, C_k$ with $|C_i| \geq \ell_i$ for all $i$.
\end{example}

\begin{example}
    Similarly to lower bounds we can also place upper bounds (capacities) on the cluster sizes. That is, given upper bounds $u_1, \ldots, u_k$, we only allow clusterings $C_1, \ldots, C_k$ with $|C_i| \leq u_i$ for all $i$.
\end{example}

\begin{example}
    In the $k$-clustering problem with $z$ outliers we are allowed to ignore up to $z$ points from the input space (these do no factor into the cost of the clustering). Although this is not a clustering constraint per se, we can model this as a $(k+z)$-clustering problem with the restriction that in any valid clustering at least $z$ clusters have to consist of exactly one point.
\end{example}

These are only a few examples, some more follow in Section~\ref{sec-other-constraints}.
Clustering constraints complicate the task of producing a clustering from a given set of centers. Simply assigning every point to its closest center (as done in unrestricted k-clustering problems and suggested above) might not be possible since this can yield an invalid clustering. However, for a lot of important clustering constraints (e.g. all those mentioned in this paper) an optimal assignment of points to centers can actually be computed in polynomial time using min-cost-flow algorithms (see for example~\cite{DX15}).
For the special case that $k$ is a constant, Ding and Xu~\cite{DX15} propose a general method to deal with Euclidean clustering problems with constraints. They show how to compute a list of possible solutions such that with constant probability, at least one solution in the list provides a good solution for the $k$-means problem under any given desired constraint. They model solutions as center sets, so the list contains sets of $k$ centers. The idea is that for any clustering, there is at least one set of $k$ centers in this list which is approximately as good as using the optimal centers (the centroids) for all clusters.
Bhattacharya et. al.~\cite{BJK18} improved this result by providing a smaller list of solutions.

\begin{theorem}[cf. Theorem 1 in \cite{BJK18}]\label{thm-bhat}
    Let $P \subset \mathbb{R}^d$ be a finite set of $n$ points, $k$ a positive integer and $\varepsilon > 0$ some error parameter.
    There is a randomized algorithm that computes in $O\left(n d 2^{\mathcal{\tilde{O}}\left(\frac{k}{\varepsilon}\right)}\right)$ time a list $\mathcal{L}$ of size $2^{ \mathcal{\tilde{O}}\left(\frac{k}{\varepsilon}\right)}$ with the following property. For every clustering $C_1, \ldots, C_k$ there exists with constant probability a set of centers $C = \{c_1, \ldots, c_k\} \in \mathcal{L}$, such that
    \[ \sum_{i = 1}^k \cost(C_i, c_i) \leq (1 + \varepsilon) \sum_{i = 1}^k \cost(C_i, \mu(C_i)).\]
\end{theorem}

So Theorem~\ref{thm-bhat} says that for every clustering, there is a good solution contained in $\mathcal{L}$. A constraint is identified with the subsets of clusterings that it allows. Now if the clustering constraint at hand admits an efficient optimal assignment algorithm as we outlined above, then applying this algorithm to every solution in the list and taking the one with lowest cost yields a probabilistic PTAS.

For example, if we have $k+z$ centers but $z$ are designated to be outliers, then we can model this by setting up a flow network where the points are the sources (with demand $1$ each) and the centers are the sinks (with an unbounded demand for the normal centers and a demand of $1$ for the outliers). Similar constructions work for upper and lower bounds on the number of points, and also for fairness constraints~\cite{DX15,SSS19}.

Going through all solutions in the list and applying this method to find the one with lowest cost gives a $(1+\varepsilon)$-approximation.

In this paper, we show that when $d$ is small, a similar result can be obtained by using known techniques from the area of \emph{coresets} for $k$-clustering problems (and in particular, by applying a recent paper by Schmidt, Schwiegelshohn and Sohler~\cite{SSS19}). The advantage is that this approach directly implies streaming algorithms and also parallel computation. The drawback is the exponential dependency on the dimension.

A streaming algorithm only reads the data once and can store very little in memory. Both in theory and practice, the most popular tool to achieve this is to use the pass over the data to compute a (weighted) summary of the points and then solve the problem on the summary in the end (or on demand, during the stream).
A \emph{coreset} is such a weighted summary which satisfies a very strong approximation guarantee: For a coreset $S$, it is guaranteed that for \emph{any} solution to the $k$-clustering problem, i.e., for any set of $k$ centers $C$, clustering the points in $S$ with $C$ costs $(1\pm \varepsilon)$ times the cost of clustering the original points with $C$. Since this is true for any possible solution, running an approximation algorithm on $S$ yields an approximative solution for the original data, too.

\begin{definition}[Coreset~\cite{HPM04}]\label{def:coreset}
    A set $S \subseteq \mathbb{R}^d$ together with non-negative weights $w:S \to \mathbb{N}$ is a $(k,\varepsilon)$-coreset for a point set $P\subseteq \mathbb{R}^d$ 
    if for every set $C
        \subseteq \mathbb{R}^d$ of $k$ centers we have
    $(1-\varepsilon) \cdot \cost(P,C) \le  \cost(S,C) \le (1+\varepsilon) \cdot \cost(P,C)$.
\end{definition}

The definition requires that for \emph{every} possible solution, so for every set of $k$ centers $C$, the cost of $P$ with $C$ is approximated by the (weighted) cost of the points in $S$ with the centers $C$. This is a fairly strong definition and it implies that any $\alpha$-approximation algorithm (that can handle weighted points) can be run on the coreset in order to obtain an $\alpha(1+O(\varepsilon))$-approximation.

This is not the whole reason for the success of coresets. Indeed, it turns out that coresets for $k$-clustering problems obeying Definition~\ref{def:coreset} have another very convenient property: They are \emph{mergeable} in an oblivious way. This means that if $(S_1,w_1)$ is a coreset for a set $P_1$ and $(S_2,w_2)$ is a coreset for a set $P_2$, the the union $S_1 \cup S_2$ (with concatenated weight functions) is a coreset for $P_1 \cup P_2$. This is true because $k$-clustering cost functions are linear; the cost of $P_1 \cup P_2$ is the sum of the cost of $P_1$ and the cost of $P_2$, for any set of centers.

This mergeability makes it possible to use the so-called merge-and-reduce technique for coreset computations~\cite{BS80}. The idea behind this technique is to read the stream in blocks, turn every block into a coreset, and union coresets (\lq merge\rq) and then \lq reduce\rq\ them by again calling the coreset algorithm, following a computation structure that resembles a complete binary tree. The merge-and-reduce framework allows to convert a coreset algorithm into a streaming algorithm that can output a coreset for the whole data set (the resulting coreset will be a bit larger than the offline version).

There are several approaches to construct coresets, and the smallest sizes for $k$-median, $k$-means and many other $k$-clustering variants are currently achieved by using the work of Feldman and Langberg~\cite{FL11}. However, in this paper we point out the universality of an older approach to construct coresets that we call \emph{movement-based}. Movement-based constructions are for example given by Har-Peled and Mazumdar~\cite{HPM04}, Frahling and Sohler~\cite{FS05}, and Fichtenberger et. al.~\cite{FGSSS13} (despite not being named movement-based in these works).

\begin{definition}
    \label{def:movement_based}
    A weighted set $(S, w)$ is a movement-based $(k, \varepsilon)$-coreset for $P$ if
    \begin{enumerate}
        \item $w$ maps to integer values,
        \item there exists a mapping $\pi: P \to S$ with $|\pi^{-1}(s)| = w(s)$ for all $s \in S$ that satisfies
              \[ \sum_{x \in P} \dist(x, \pi(x))^m \leq \left(\frac{\varepsilon}{2m}\right)^m \opt_P.\]
    \end{enumerate}
    The expanded version $S'$ of $S$ is a multi set that contains $w(s)$ copies of every point $s \in S$.
\end{definition}

The naming is inspired by imagining that the coreset is produced by moving points around such that multiple points coincide; these are then replaced by a weighted point. The specific condition stems from a lemma that is used in the literature on movement-based coreset constructions and which we prove in Appendix~\ref{movement-appendix} in the following fairly general form.

\begin{restatable}{lemma}{lemMovementCoreset}
    \label{lem:movement_based_is_coreset}
    Let $C \subseteq X$ be a set of $k$ centers and $\varepsilon \in (0,1]$ be a real number.
    If there exists a bijection $\pi: P \to Q$ of finite subsets $P, Q \subseteq X$, such that
    $ \sum_{x \in P} \dist(x, \pi(x))^m \leq \left(\frac{\varepsilon}{2m}\right)^m \cdot \cost(P,C)$,
    then it holds that
    \[ | \cost(P,C) - \cost(Q,C) | \leq \varepsilon \cdot \cost(P,C). \]
\end{restatable}

If the cost of moving the points along $\pi$ is upper bounded by $\left(\frac{\varepsilon}{2m}\right)^m \cdot \opt_P$, then we get a coreset (see Appendix~\ref{movement-appendix}). 
This justifies Definition~\ref{def:movement_based}.

\begin{restatable}{corollary}{corMovBasedCoreset}
    \label{cor:movement_based_is_coreset}
    If $(S, w)$ is satisfies Definition \ref{def:movement_based}, then $(S, w)$ is a $(k, \varepsilon)$-coreset for $P$.
\end{restatable}

As indicated above, coresets can be used to obtain streaming algorithms, and indeed, movement-based coreset constructions that work in the streaming setting are known.

\begin{theorem}[Known results on movement-based coresets]
    \label{thm:known-sizes-coresets}
    Let $P \in \mathbb{R}^d$ be a point set given as a stream in the insertion-only model, let $0 < \varepsilon \le 1$, let $k \in \mathbb{N}$  and let the dimension $d$ be constant. A movement-based $(k, \varepsilon)$-coreset of size
    \begin{itemize}
        \item ${O}(k \varepsilon^{-d} \log^2 n)$ for $k$-median
        \item $\tilde{O}(k \varepsilon^{-d} \log n)$ for $k$-means
    \end{itemize}
    can be constructed time that is polynomial in $n$, $k$ and the spread $\log \Delta$ of $P$ ($\Delta$ is the maximum pairwise distance divided by the minimum pairwise distance) and exponential in $d$.
\end{theorem}
\begin{proof}
    The result for $k$-means is from Fichtenberger et. al.~\cite{FGSSS13}. Although that algorithm probably also works for $k$-median, it is not covered there, but the original paper by Har-Peled and Mazumdar~\cite{HPM04} covers both. The size stated in the theorem is from the updated version of the paper~\cite{HPM18}.
\end{proof}

Now for constraints, a subtle thing happens: While Lemma~\ref{lem:movement_based_is_coreset} is still true and can still be used to obtain a coreset, these coresets are not longer mergeable. They thus lose their main benefit for obtaining streaming (and distributed) algorithms.

To see this, consider for example \emph{lower bounded} $k$-clustering. Given a point set $P$ and two numbers $k$ and $\ell$, choose $k$ centers and assign all points to a center in such a way that every chosen center gets at least $\ell$ points assigned. 
Imagine a one-dimensional point set with $\ell$ points at $0$ and $\ell$ points at $1$, and set $k=2$. The optimum solution puts one center at $0$ and one at $1$ and costs nothing. A coreset with the spirit of Definition~\ref{def:coreset} can also easily be found: We store a point at $0$, one at $1$, and give each point weight $\ell$. This is virtually equivalent to the original point set and of size two.

However, assume that we do not get the whole point set at once, but we get two different point sets $P_1$ and $P_2$, each containing $\ell/2$ points at either location (let's say $\ell$ is even). Suppose we are supposed to compute a coreset for each set individually and obtain a coreset for $P$ as by merging the two coresets.

Now the crucial observation is that for $P_1$ individually, any lower-bounded clustering is very expensive: Since there are only $\ell /2$ points in $P_1$, we can only open one center even though $k=2$. So no matter what we do, we pay at least $(\ell/2) \cdot (1/4)$. This means that if we follow the spirit of Definition~\ref{def:coreset} -- demanding that we disturb the cost by at most an $\varepsilon$-fraction -- then we are actually allowed a huge error. We could compute a \lq coreset\rq\ with an error of around $\varepsilon \ell$. Now if we compute such \lq coresets\rq\ for $P_1$ and $P_2$ and merge them, we get a \lq coreset\rq\ for $P$ with error proportional to $\varepsilon \ell$. But if we union $P_1$ and $P_2$, the optimum cost becomes zero! Thus, the error is way too large, we did not actually get a coreset for $P$.

So the bottom line is that the union of $S_1$ and $S_2$ is not a coreset for $P_1 \cup P_2$ because the cost function \emph{can decrease} when we merge point sets. This happens in a similar fashion for other constraints, for example for \emph{fair $k$-means} as outlined in~\cite{SSS19}. However,~\cite{SSS19} provides a way to deal with this problem in the form of a stronger coreset definition. We will demonstrate that ideas from~\cite{SSS19} can be applied much more broadly to obtain mergable coresets for $k$-clustering problems under constraints.

\paragraph*{Closely related work}

This paper builds heavily on the work of Schmidt, Schwiegelshohn and Sohler~\cite{SSS19} where it is observed that movement-based coresets achieve mergeable coresets for clustering under specific fairness constraints. We demonstrate that the techniques developed there can be applied to a much larger class of constrained clustering problems. This follows up upon a recent line of research initiated by Ding and Xu~\cite{DX15} and further refined by Bhattacharya~\cite{BJK18} to obtain polynomial-time approximation schemes for clustering with constraints.
The coreset approach can also achieve approximation schemes yet with the benefit that they can be executed in streaming and be used in distributed settings.
This works by using coresets which have a long history. We use movement-based coreset constructions as developed in~\cite{HPM04,FS05,FGSSS13}. This means that the coreset sizes depend exponentially on $d$ (a drawback compared to~\cite{BJK18}). An advantage is that movement-based coreset constructions also have available and fairly efficient implementations (see for example \cite{FGSSS13}) which only need very minor adaptations to work for the constrained case.

For fair clustering, more coreset constructions are known. Independently of~\cite{SSS19},  Huang, Jiang and Vishnoi~\cite{HJV19} develop a coreset construction for fair clustering which has a size that is exponential in the dimension. Bandyapadhyay, Fomin and Simonov~\cite{BFS20} show that it is possible to avoid the exponential dependence on $d$ and give a coreset of size $\poly(\log n, \epsilon, d, t)$, where $t$ is the number of different colors in the input.

The best known constructions for unconstrained clustering are derived from the works by Feldman and Langberg~\cite{FL11}.
Huang et. al.~\cite{HJLW18} develop a coreset for $k$-median in general (doubling) metrics with outliers but not in the streaming setting. Coresets were also recently used to get better bounds on approximating capacitated $k$-median and $k$-means in an FPT-setting by Cohen-Addad and Li~\cite{CAL19}. This work achieves a $(3+\varepsilon)$-approximation for capacitated $k$-median and a $(9+\varepsilon)$-approximation for capacitated $k$-means in time that is polynomial if the number of centers $k$ is a constant (but the result is not in a streaming setting). Coresets for constrained clustering also follow from the construction in this paper, yet the movement-based constructions here yield larger coreset sizes which depend exponentially on the doubling dimension of the underlying metric space.

\section{Coresets for lower bounds, upper bounds and outliers}

We start by considering prominent constraints that restrict the number of points in clusters. We identify such constraints by a set of vectors. Each vector describes one allowed clustering in terms of sizes, i.e., it specifies the number of points we want in every cluster in the partitioning.

\begin{definition}
    A \emph{size constraint} for $P$ is a set of vectors $\mathbf{K} \subseteq \mathbb{N}^k$ with $\|K\| = |P|$ for all $K \in \mathbf{K}$. The vectors describe which cluster sizes are allowed.
    An assignment $\alpha: P \to C$ to a set of $k$ centers $C = \{c_1, \ldots,c_k\}$ satisfies the constraint $\mathbf{K}$ if there exists a vector $K \in \mathbf{K}$ such that $|\alpha^{-1}(c_i)| = K_i$ for all $i = 1, \ldots, k$. That is, there have to be exactly $K_i$ points in cluster $i$.
\end{definition}

This definition covers the three constraints we mentioned above.

\begin{lemma}
    \label{lem:lower_bounds_are_size_constraints}
    Lower bound and upper bound constraints as well as outliers are special cases of size constraints.
\end{lemma}

\begin{proof}
    Say we are considering a lower bound constraint with values $\ell_1, \ldots \ell_k$. Then this the same as imposing the size constraint
    \[ \{ K \in \mathbb{N}^k \; | \; \ell_i \leq K_i \leq n \text{ for all } i = 1, \ldots, k\}.\]
    Similarly upper bound constraints with values $u_1, \ldots, u_k$ can be described by the size constraint
    \[ \{ K \in \mathbb{N}^k \; | \; 0 \leq K_i \leq u_i \text{ for all } i = 1, \ldots, k\}.\]
    As outlined earlier, the $k$-clustering problem with $z$ outliers can be viewed as the $(k+z)$-clustering problem with the restricion, that $z$ clusters have to consist of exactly one element. Thus we can just impose the size constraint
    \[ \{ K \in \mathbb{N}^{k+z} \; | \; K_i = 1\text{ for all } i = 1, \ldots, z\}\]
    to get the restriction.
\end{proof}

The description of a constraint in this form may have a very large size because we basically enumerate all feasible combinations of cluster sizes. However, we will only need this description in proofs. Now we define a general size constrained clustering problem.

\begin{definition}
    In the size constrained $k$-clustering problem we are additionally given a size constraint $\mathbf{K} \subset \mathbb{N}^k$ for $P$. For an assignment $\alpha: P \to C$ to a set of $k$ centers $C \subseteq X$ define its cost as
    \[ \cost(\alpha) = \sum_{x \in P} \dist(x, \alpha(x))^m.\]
    The goal is to find a set of $k$ centers $C \subseteq X$ that minimizes
    \[ \cost_\mathbf{K}(P, C) = \min_\alpha \cost(\alpha),\]
    where $\alpha$ ranges over all assignments of $P$ to sets of $k$ centers that satisfy $\mathbf{K}$.
\end{definition}

What we want is a summary of the input points that is mergeable and approximately preserves the cost of an optimal assignment. This place will be taken by size coresets. Essentially, they preserve the cost of an optimal valid assignment under all possible size constraints.

\begin{definition}
    A $(k, \varepsilon)$-size coreset for $P$ is a set $S$ with integer weights $w: S \to \N$ that satisfies
    \[ |\cost_\mathbf{K}(P, C) - \wcost_\mathbf{K}(S, C)| \leq \varepsilon \cost_\mathbf{K}(P, C)\]
    for all size constraints $\mathbf{K}$. In the above inequality we set
    $\wcost_\mathbf{K}(S, C) = \cost_\mathbf{K}(S', C)$,
    where $S'$ is the expanded version of $S$. That is, a point $s \in S$ is treated as $w(s)$ unit points that can all be assigned to different clusters.
\end{definition}

Size coresets are a simplification of color coresets from~\cite{SSS19} (described in Section~\ref{sec-other-constraints}).  They share the property that they are mergeable (proof in Appendix~\ref{sec-size-coresets-mergeable}).

\begin{restatable}{theorem}{thmMergeableSizeConstraints}
    \label{thm:mergeable_for_size_constraints}
    Let $(S_1, w_1)$ and $(S_2, w_2)$ be $(k, \varepsilon)$-size coresets for finite subsets $P_1$ and $P_2$ of $X$, respectively. Then then the merged set $(S, w) = (S_1 \cup S_2, w_1 + w_2)$ is a $(k,\varepsilon)$-size coreset for $P = P_1 \cup P_2$.
\end{restatable}

The benefit of Theorem~\ref{thm:mergeable_for_size_constraints} is that for \emph{any} construction of size coresets, we would directly know that they are mergeable and thus allow the usage of merge-and-reduce to obtain streaming algorithms. However, we do not make use of this here since we use movement-based coresets, for which streaming algorithms are already known. So it remains to show that movement-based coresets are indeed size coresets. The proof is similar to~\cite{SSS19}.

\begin{theorem}
    \label{thm-size-constraints}
    Let $\varepsilon \in (0, 1)$. If $(S,w)$ is a movement-based $(k, \varepsilon)$-coreset for $P$, then $(S,w)$ is also a size coreset for $P$.
\end{theorem}

\begin{proof}
    Let $S'$ be the expanded version of $S$. By assumption there exists a bijection $\pi: P \to S'$, such that
    \[ \sum_{x \in P} \dist(x, \pi(x))^m \leq \left(\frac{\varepsilon}{2m}\right)^m \cdot \opt_P, \]
    We show that merging an optimal valid assignment $\gamma: P \to C$ with $\pi^{-1}$ yields a valid assignment $\gamma \circ \pi^{-1}: S' \to C$ which doesn't cost much more.
    Again enumerate $P = \{ p_1, \ldots, p_n \}$ and form the vectors
    \[ v_c = (\dist(p_1, \gamma(p_1)), \ldots, \dist(p_n, \gamma(p_n)))^\top\]
    and
    \[ v_p = (\dist(p_1, \pi(p_1)), \ldots, \dist(p_n, \pi(p_n)))^\top.\]
    Then $\|v_c\|_m = \cost_\mathbf{K}(P,C)^\frac{1}{m}$ and $\|v_p\|_m \leq \frac{\varepsilon}{2m} (\opt_P)^\frac{1}{m} \leq \frac{\varepsilon}{2m}\cost(P,C)^\frac{1}{m}$ so that
    \begin{align*}
        \wcost_\mathbf{K}(S,C) & \leq \sum_{y \in S'} \dist(y, \gamma(\pi^{-1}(y)))^m = \sum_{x \in P} \dist(\pi(x), \gamma(x))^m \\
                               & \leq  \sum_{x \in P} \left(\dist(x, \gamma(x)) + \dist(x, \pi(x)) \right)^m
        = \|v_c + v_p\|_m^m \leq \left(\|v_c\|_m + \|v_p\|_m\right)^m                                                             \\
                               & = \sum_{i = 1}^m \binom{m}{i} \|v_c\|_m^{m-i} \|v_p\|_m^i
        \leq \sum_{i = 1}^m \binom{m}{i} \cost_\mathbf{K}(P,C)^\frac{m-i}{m} \cdot \frac{\varepsilon}{2m} \cost(P,C)^\frac{i}{m}  \\
                               & \leq \sum_{i = 1}^m \binom{m}{i} \cdot \frac{\varepsilon}{2m} \cdot \cost_\mathbf{K}(P,C)
        = \left(1 + \frac{\varepsilon}{2m}\right)^m \cdot \cost_\mathbf{K}(P,C)                                                   \\
                               & \leq (1 + \varepsilon) \cdot \cost_\mathbf{K}(P,C)
    \end{align*}
    In the above argumentation we can swap $S'$ and $P$ to get $\cost_\mathbf{K}(P, C) \leq (1+\varepsilon) \cdot \wcost_\mathbf{K}(S, C)$. Since $1 - \varepsilon \leq \frac{1}{1+\varepsilon}$ this shows that $(1 - \varepsilon) \cdot \cost_\mathbf{K}(P,C) \leq \wcost_\mathbf{K}(S, C)$ also holds.
\end{proof}

Notice that the argumentation in the proof of Theorem~\ref{thm-size-constraints} can be applied to any other clustering constraint that just restricts the space of allowed clusterings. We use this in the next section.

As a combination of Lemma \ref{lem:lower_bounds_are_size_constraints} and Proposition \ref{thm:mergeable_for_size_constraints} we get the following result.

\begin{corollary}
    Movement-based coresets are mergeable coresets for the constrained variants of lower bounds, upper bounds and outliers.
\end{corollary}

\section{Coresets for other clustering constraints}\label{sec-other-constraints}

The previous section showed using a simplified analysis of~\cite{SSS19} that movement-based coresets can be used when dealing with size constraints. In this section it will be shown that~\cite{SSS19} actually also works for other clustering variants.
To explain it, we need a following generalization of $k$-clustering. It will be used to model constraints, however, it is not as intuitive to use as size constraints are. The generalization is based on assigning \emph{colors} to all points in $P$. The motivation for this comes from fair clustering, where colors are used to model sensitive attributes. Basically, the constraints we will now cover are about the distribution of attributes over clusterings. Assuming that points can be colored allows to model constraints beyond the fairness constraints that were considered in~\cite{SSS19}.

So from now on we assume that we additionally get a coloring $f: P \rightarrow \{1,\dots,\ell\}$ of our point set $P$ with $\ell$ colors. A new type of constraint introduced in~\cite{SSS19} are color constraints. Basically these are size constraints for points of the same color. That is, a color constraint declares which color distributions of clusterings are considered to be valid. It tells exactly how many points of color $j$ are to be contained in cluster $i$.

\begin{definition}
    \label{def:color-constraint}
    A \emph{color constraint} for $P$ is a set of matrices $\mathbf{K} \subset \mathbb{N}^{k \times \ell}$ with $\sum_i K_{ij} = |f^{-1}(j)|$ for all $j = 1, \ldots, \ell$. A clustering $C_1, \ldots, C_k$ of $P$ satisfies the color constraint $\mathbf{K}$ if there exists a matrix $K \in \mathbf{K}$ such that
    \[ |f^{-1}(j) \cap C_i| = |\{x \in C_i \; | \; f(x) = j\}| = K_{ij}.\]
\end{definition}

Just like in the previous section with size constraints we may pose a variant of $k$-clustering based on these new color constraints. Additionally to a colored input set we get a color constrait and are supposed to find an assignment of minimum cost that satisfies the constraint. Obviously the color constraint can be quite large, but they do not actually have to be part of the input.

\begin{definition}
    \label{def:color-constrained-clustering}
    In a color constrained $k$-clustering problem we are additionally give a color constraint $\mathbf{K}$ for the colored set $P$ and the goal is to find a set of $k$ centers $C \subseteq X$ that minimizes
    \[ \cost_{\mathbf{K}}(P, C) = \min_\alpha \sum_{x in P} \dist(x, \alpha(x))^m, \]
    where $\alpha$ ranges over all assignments of $P$ to some set of $k$ centers that satisfy $\mathbf{K}$.
\end{definition}

The definition of movement-based coresets from the previous section can naturally incorporate colors. Basically one makes it a multiset, so that it can contain up to $\ell$ copies with different weights for every point. That is, if we view the movement-based coreset as the result of merging nearby points, then this time around we save for every color the numbers of points of that color that were merged.

\begin{definition}
    \label{def:movement_based_for_colors}
    A weighted multiset $(S, w)$ with colors $f': S \to \{1, \ldots, \ell\}$ is a movement-based $(k, \ell, \varepsilon)$-color coreset for the color input set $P$ if
    \begin{itemize}
        \item the weights $w$ take on positive integer values,

        \item there exists a mapping $\pi: P \to S$ such that
              \[ \sum_{x \in P} \dist(x, \pi(x))^m \leq \left(\frac{\varepsilon}{2m}\right)^m \opt_P\]
              holds and $|\pi^{-1}(s) \cap f^{-1}(i)| = w(s)$ for all points $s \in S$ of color $i$ $(1 \leq i \leq \ell)$.
    \end{itemize}
\end{definition}

Similarly to movement-based coresets from the previous section these assumptions imply a certain type of coreset defined \cite{SSS19} and defined below.

\begin{definition}[\cite{SSS19}]
    \label{def:color-coreset}
    A non-negatively integer weighted set $S \subseteq \mathbb{R}^d$ with a coloring $f':S \rightarrow \{1,\dots,\ell\}$ is a $(k,\varepsilon)$-color-coreset for $P$, if for every set $C \subseteq \mathbb{R}^d$ of $k$ centers and every color constraint $K$ we have
    \[
        |\wcost_\mathbf{K}(S,C) - \cost_\mathbf{K}(P,C)| \leq \varepsilon \cdot \cost_\mathbf{K}(P,C).
    \]
\end{definition}

Essentially, a color constraint approximately preserves the cost of an optimal assignment given a arbitrary set of $k$ centers and an arbitrary color assignment. If there is only one color involved, then these are just the same as size coresets. It is shown in ~\cite{SSS19} that color coresets are mergeable for $k$-means.

\begin{corollary}
    Color coresets are mergeable.
\end{corollary}

Just like with movement-based coresets in the previous section one can show that movement-based $(k, \ell, \varepsilon)$-color coresets are coresets that respect all color constraints involving $\ell$ colors. This is proven in~\cite{SSS19} for $k$-means, and the proof for the general case is basically the same as the proof of~Theorem\ref{thm-size-constraints}.

\begin{corollary}
    A movement-based color coreset is a color coreset.
\end{corollary}

Importantly, any constructions that yield movement-based coresets can easily be adjusted to yield movement-based color coresets. When the original construction would merge a point set $P'$ into a point $p$ the new construction can instead just add $\ell$ copies $p$, where the weight assigned to copy $i$ is $|P' \cap f^{-1}(i)|$. This just increases the coreset size by a multiplicative factor of $O(\ell)$. In fact, the constructions mentioned in Theorem \ref{thm:known-sizes-coresets} can be adjusted in this manner, which implies  the following.

\begin{corollary}
    \label{thm:size-color-coresets}
    There are streaming constructions that yield movement-based $(k, \ell, \varepsilon)$-color coresets of size $O(\ell k \varepsilon^{-d} \log^2 n)$ for $k$-median and $O(\ell k \varepsilon^{-d} \log n)$ for $k$-means with runtimes that are polynomial in $n,k,\ell$ and the spread $\log \Delta$ of $P$ and exponential in $d$.
\end{corollary}

Now that this has been established, we see what we can do with it. First note that any clustering constraint (that restricts the space of allowed clusterings) can be modeled as a color constraint if we assume the points to be colored in some specific way. This assumption on the coloring of the input space is not unreasonable. In fact, this naturally arises in several practical applications. More will be said about this further down.

\begin{observation}
    Every clustering constraint (that restricts the space of allowed clusterings) can be encoded by a color constraint if an appropriate coloring is chosen.
\end{observation}

\begin{proof}
    If one colors all points in a different color, then a color constraint tells us for every individual point exactly where it has to be put.
\end{proof}

Obviously assuming that every point has a different color is useless and any movement-based color coreset will be just as large as the input set which it is supposed to approximate. However, if we consider any specific clustering constraint then we may be able to work with much less colors. With that let us focus on more specific constraints (also to see that the assumption of colored points is not unreasonable).

\begin{example}
    In $k$-clustering problems with must-link constraints we are given a set of links $L \subset \binom{P}{2}$ between points in $P$. The restriction imposed is that any two points with a link between them have to be placed in the same cluster. Notice that this forms a graph where all points in the same connected component have to be put into the same cluster.
    Now, one can consider all pairs of points that form a link to have the same color. Every point in a non-singular connected component thus has the color and every connected component uses a different color $1 \leq i \leq \ell - 1$. The remaining points may share the same color $\ell$. Then the must-link constraint can just be described as the color constraint $\mathbf{K}$ that consists of all matrices with at most 1 non-zero entry in the first $l-1$ columns. The number of colors can be bounded by the number of must-link constraints (and can be smaller if there are many connected components).
\end{example}

Depending on the application such a coloring of the points may already be given as an encoding of different properties/attributes that was additionally measured. The links may then just describe various compatibility constraints of these colors (read properties). That two points must be placed in the same cluster does not necessarily have to be due to some singular reason that only concerns these exact two points. It might rather be the result of the points sharing specific colors (attribute). Colors are a much more direct and effective way of encoding such information. An algorithm used in practice that goes through every cannot-link to check whether a given clustering violates it is unlikely. Usually there should be a more efficient procedure so that some of the information contained in the cannot-links can be abstracted from.

Must-link constraints are still relatively simply to hande, but cannot-link constraints require more effort and may need more colors.

\begin{example}
    Cannot-link constraints are similar. A cannot-link constraint is also a set of links $L \subset \binom{P}{2}$ between points in $P$, but this time they tell us which points cannot lie in the same cluster. If there is a link between two points they have to lie in different clusters. Consider two points that share a link to have different colors and the remaining (not linked) points to share a color. Let $\mathbf{K}$ consist of all matrices $K \in \mathbb{N}^{k \times l}$ that satisfy the following: If there is a cannot-link between a point of color $i$ and a point of color $j$, then in every row of $K$ the $i$-th or the $j$-th position has to be zero. The cannot-link constraint then is completely encoded in $\mathbf{K}$.
\end{example}

The coloring $f$ outlined above is almost certainly not optimal, especially, when the cannot-links are to be considered to be between groups of points and not just between individual points. Again, such a coloring might already be provided in the form of attributes. If the cannot-links are transitive then the constraints coincides with the chromatic constraint introduced in \cite{DX11} and described further below.

There are also quite a few constraints that are directly based on colors (attributes) and their distributions. All of these can be described by color constraints and so color coresets may find their use.
One example which is extensively discussed in~\cite{SSS19} are fairness constraints.
Another special variant of color constraints include the chromatic constraints introduced by Ding and Xu~\cite{DX11}.

\begin{example}\label{example:chromatic}
    Chromatic constraints model a behavior where points of the same color are incompatible with one another. A clustering $C_1, \ldots, C_k$ is valid, if
    \[ |f^{-1}(j) \cap C_i| = 1 \]
    for all $i$ and $j$. This is just the same as imposing the color constraint that consists of all matrices $K$ with entries in $\{0, 1\}$ only.
\end{example}

As a last example consider the $l$-diversity constraints mentioned in \cite{DX15}.

\begin{example}\label{example:diversity}
    A clustering $\{C_1, \ldots, C_k\}$ of $P$ satisfies the $l$-diversity constraint, if
    \[ \frac{|f^{-1}(j) \cap C_i|}{|C_i|} \leq \frac{1}{l} \]
    for all colors $j$ and all indices $i$. This restriction is the same as imposing the color constraint $\mathbf{K}$ that consists of all matrices $K$ with
    \[ \frac{K_{ij}}{\sum_{j'} K_{ij'}} \leq \frac{1}{l}\]
    for all $i$ and $j$.
\end{example}

Since all of the mentioned constraints can be modeled as color constraints, it is possible to compute a color coreset for any of the corresponding constrained $k$-median/$k$-means problems by using Corollary~\ref{thm:size-color-coresets}. In the final section, we give (standard) pointers on how to use the resulting coreset to approximate the constrained clustering problem.

\section{Algorithms for constrained clustering problems}

A coreset allows us to transfer approximation algorithms to the streaming setting. Any approximation algorithm for a constrained clustering problem that can handle weighted points can be run on the coreset to obtain an approximative solution for $p$. Here is a standard way to get an approximation result.

\begin{lemma}
    Let $\varepsilon \in (0,1)$.
    Computing an $\alpha$-approximation for a color constrained $k$-clustering problem on a $(k,\varepsilon/3)$-color coreset $(S,w,f')$ for input $(P,f,\mathbf{K},k)$ yields an $(\alpha+\varepsilon)$-approximation.
\end{lemma}
\begin{proof}
    Let $C^\ast$ be the optimum solution for the color constrained $k$-clustering problem on $(P,f,\mathbf{K},k)$ and set $\varepsilon'=\varepsilon/3$. By Definition~\ref{def:movement_based_for_colors},
    $\wcost_\mathbf{K}(S,C^\ast) \le  (1+\varepsilon')\cost_\mathbf{K}(P,C^\ast)$, so this is an upper bound on the cost of an optimal solution on the coreset. This means that an $\alpha$-approximate solution $C^+$ on $(S,w,f',\mathbf{K},k)$ satisfies that $\wcost_\mathbf{K}(S,C^+) \le \alpha (1+\varepsilon')\cost_\mathbf{K}(P,C^\ast)$. By Definition~\ref{def:movement_based_for_colors},
    \[
        \cost_\mathbf{K}(P,C^+) \le \wcost_\mathbf{K}(S,C^+) + \varepsilon' \cdot \cost_\mathbf{K}(P,C^+)
        \Leftrightarrow
        \cost_\mathbf{K}(P,C^+) \le \frac{1}{1-\varepsilon'} \cdot \wcost_\mathbf{K}(S,C^+),
    \]
    so we have that $\cost_\mathbf{K}(P,C^+) \le \frac{1}{1-\varepsilon'} \cdot\alpha (1+\varepsilon')\cost_\mathbf{K}(P,C^\ast)$, which is at most $\alpha(1+3\varepsilon')\cost_\mathbf{K}(P,C^\ast)$ when assuming $\varepsilon' \le 1/2 \leftrightarrow \varepsilon \le 3/2$ which is clearly satisfied. Since $3\varepsilon'=\varepsilon$, the claim is proven.
\end{proof}

There are not overly many constant-factor approximation algorithms for constrained $k$-clustering known. For example, an algorithm for lower-bounded facility location by Svitkina~\cite{S10} can be probably adapted to work for $k$-median (as stated by Ahmadian and Swamy~\cite{AS16}). After doing this adaptation and also adapting it to work for weights, one can use it together with the coreset to obtain a constant-factor approximation in very fast time in the streaming model. (This comes at the cost of an exponential dependence on $d$ in the space requirement and running time). Another algorithm providing a constant-factor approximation is for $k$-median with outliers due to  Krishnaswamy, Li and Sandeep~\cite{KLS18}.

A fairly general method to obtain an approximation from the coreset is to design a polynomial approximation scheme (PTAS) which iterates through all possible solutions on the coreset. This introduces an exponential dependence on the number of centers $k$ (in addition to the exponential dependence on $d$ which movement-based coresets imply), yet at least it is applicable to all possible constraints. We outline such an algorithm for the case of $k$-means in Section~\ref{appendix:inaba} by using a well-known PTAS technique for coresets going back to Inaba et al. \cite{IKI94} which works in the general setting of constraints. Similarly to Theorem~\ref{thm-bhat}, one needs an efficient algorithm that computes an assignment of the input points to a test set of centers. Of course, alternatively, one can also work to adapt one of the proposed PTASes for constrained clustering~\cite{BJK18,DX11} to weighted points and combine it with the coreset to obtain the same result.

\begin{restatable}{theorem}{ptasInaba}
    There exists a PTAS (for $\varepsilon \in (0,1)$) which works in the streaming model for color constrained $k$-means clustering problems with $\ell$ colors that admit efficient optimal assignment algorithms, assuming that $d$ and $k$ are constant.
\end{restatable}

\paragraph*{Acknowledgements}

This research was supported by the Deutsche Forschungsgemeinschaft, DFG, project number SCHM 2765/1-1. 
We thank anonymous referees for their helpful feedback and pointing out typos.

\bibliography{references}

\appendix

\section{Movement-based constructions yield coresets}\label{movement-appendix}

\lemMovementCoreset*
\begin{proof}
    Enumerate $P = \{ p_1, \ldots, p_n \}$ and $Q = \{ q_1, \ldots, q_n\}$ such that
    $\pi(p_i) = q_i$ holds for all $1 \leq i \leq n$ and define
    the two vectors
    \[ v_P = (\dist(p_1, C), \ldots, \dist(p_n, C) )^\top\]
    and
    \[ v_Q = (\dist(q_1, C), \ldots, \dist(q_n, c) )^\top.\]
    The difference $|\dist(p_i, C) - \dist(q_i, C)|$ of the entries is at most $\dist(p_i, q_i)$ and so
    \begin{align*} \|v_Q - v_P\|_m^m & = \sum_{i=1}^n ||\dist(p_i, C) - \dist(q_i, C)||^m                                              \\
                          & \leq \sum_{i=1}^n \dist(p_i, q_i)^m \leq \left(\frac{\varepsilon}{2m}\right)^m \cdot \cost(P,C)
    \end{align*}
    by the assumption.
    Since $\cost(P,C) = \|v_P\|_m^m$ and $\cost(Q,C) = \|v_Q\|_m^m$ we can estimate
    \begin{align*}
        \cost(Q,C)^{\frac{1}{m}} & = \|v_Q\|_m = \|v_P + \Delta\|_m \leq \|v_P\|_m + \|\Delta\|_m                                                                                                \\
                                 & \leq \|v_P\|_m + \frac{\varepsilon}{2m} \cdot \left(\cost(P,C)\right)^{\frac{1}{m}} = \left(1 + \frac{\varepsilon}{2m}\right) \cdot \cost(P,C)^{\frac{1}{m}}.
    \end{align*}
    Raising both sides to the exponent $m$ and using the inquality $\left(1 + \frac{\varepsilon}{2m}\right)^m \leq e^{\frac{\varepsilon}{2}} \leq 1 + \varepsilon$ for $\varepsilon \in (0,1]$ yields
    \[ \cost(Q,C) - \cost(P,C) \leq \varepsilon \cdot \cost(P,C).\]
    To show that $\cost(P,C) - \cost(Q,C) \leq \varepsilon \cdot \cost(P,C)$ also holds just reverse the roles of $P$ and $Q$ in the above proof to get
    \[ \cost(P,C) - \cost(Q,C) \leq \varepsilon \cdot \cost(Q,C).\]
    If $\cost(P,C) < \cost(Q,C)$ the left side is negative and must be smaller than $\varepsilon \cdot \cost(P,C)$.
    Otherwise $\cost(Q,C) \leq \cost(P,C)$ and so
    \[ \cost(P,C) - \cost(Q,C) \leq \varepsilon \cdot \cost(Q,C) \leq \varepsilon \cdot \cost(P,C).\]
\end{proof}

\corMovBasedCoreset*
\begin{proof}
    Since the weights of $(S,w)$ are positive integers we can form the expanded multiset $S'$ that contains every point $s \in S$ exactly $w(s)$ times.
    Naturally $\pi: S \to P$ then gives rise to a map $\pi': S' \to P$ with
    \[ \sum_{x \in P} \dist(x, \pi'(x))^m \leq \left(\frac{\varepsilon}{2m}\right)^m \cdot \opt_P. \]
    Since $\opt_P \leq \cost(P, C)$ for any center set $C \subseteq X$ it follows that Lemma \ref{lem:movement_based_is_coreset} also holds for all $C$ and $(S, w)$ is a $(k, \varepsilon)$-coreset for $P$.
\end{proof}

\section{Size Coresets are mergeable}\label{sec-size-coresets-mergeable}

\thmMergeableSizeConstraints*
\begin{proof}
    Fix an arbitrary set of $k$ centers $C \subseteq X$ and let $\mathbf{K} = \{K\}$ for now consist of a single matrix $K$.
    Let $\gamma: P \to C$ be an optimal assignment satisfying $K$, so that
    \[ \cost_{K}(P, C) = \sum_{x \in P} \dist(x, \gamma(x))^m.\]
    Split up $\gamma$ into the assignments $\gamma_1: P_1 \to C$ and $\gamma_2: P_2 \to C$ by setting $\gamma_i(x) = \gamma(x)$ for $p \in P_i$. Denote with $K_1$ and $K_2$ the sizes of the clusters induced by $\gamma_1$ and $\gamma_2$, respectively.
    Since $\gamma_i$ satisfies $K_i$, we have $\cost_{K_i}(P_i, C) \leq \sum_{x \in P_i} \dist(x, \gamma_i(x))^m$ for $i = 1,2$ and so
    \begin{align*}
        \cost_{K}(P, C) & = \sum_{x \in P} \dist(x, \gamma(x))^m                                               \\
                        & = \sum_{x \in P_1}\dist(x, \gamma_1(x))^m + \sum_{x \in P_2} \dist(x, \gamma_2(x))^m \\
                        & \geq \cost_{K_2}(P_1, C) + \cost_{K_2}(P_2, C).
    \end{align*}
    On the other hand, any cheaper assignment $\gamma_i': P_i \to C$ satisfying $K_i$ could yield a cheaper assignment $\gamma': P \to C$ satisfying $K$ by taking $\gamma$ and just have it act like $\gamma_i'$ on $P_i$. As such both assignment $\gamma_1$ and $\gamma_2$ are already optimal and we have that
    \[ \cost_{K}(P, C) = \cost_{K_1}(P_1, C) + \cost_{K_2}(P_2, C). \]
    Similarly, we can argue that $\wcost_{K}(S, C) \leq \wcost_{K_1}(S_1, C) + \wcost_{K_2}(S_2, C)$. Note that this is not necessarily an equality, since an optimal assignment for $S$ satisfying $K$ does not necessarily split up into the sizes $K_1$ and $K_2$ on $S_1$ and $S_2$. However, combining both yields
    \begin{align*}
        \wcost_{K}(S, C) & \leq \wcost_{K_1}(S_1, C) + \wcost_{K_2}(S_2, C)                                               \\
                         & \leq (1 + \varepsilon) \cdot \cost_{K_1}(P_1, C) + (1 + \varepsilon) \cdot \cost_{K_2}(P_2, C) \\
                         & = (1 + \varepsilon) \cost_{K}(P, C).
    \end{align*}
    We can argue similarly that $(1 - \varepsilon) \cost_{K}(P, C) \leq \wcost_{K}(S, C)$ by just splitting up $K$ into $K_1'$ and $K_2'$ according to an optimal assignment for $S$ satisfying $K$. Then
    \begin{align*}
        \wcost_{K}(S, C) & = \wcost_{K_1}(S_1, C) + \wcost_{K_2}(S_2, C)                                                  \\
                         & \geq (1 - \varepsilon) \cdot \cost_{K_1}(P_1, C) + (1 - \varepsilon) \cdot \cost_{K_2}(P_2, C) \\
                         & \geq (1 - \varepsilon) \cost_{K}(P, C)
    \end{align*}
    and we get
    \[ |\cost_\mathbf{K}(P, C) - \wcost_\mathbf{K}(S, C)| \leq \varepsilon \cost_\mathbf{K}(P, C)\]
    in the case of a singular size constraint $\mathbf{K} = \{K\}$.

    For general size constraints $\mathbf{K}$ that consist of more than one matrix we can now argue as follows.
    Let $\gamma_P: P \to C$ be an optimal assignment of $P$ satisfying $\mathbf{K}$ and $\gamma_S: S \to C$ an optimal assignment of $S$ satisfying $\mathbf{K}$. We know that $\gamma_P$ must satisfy some $K_P \in \mathbf{K}$ and so by our observation there must exist an allocation $\gamma'_S: S \to C$ satisfying $K_P$ such that $\cost(\gamma'_S) \leq (1 + \varepsilon)\cost(\gamma_P)$. The space of solutions from which $\gamma_S$ is chosen includes all allocations satisfying $K_P$ and so
    \[ \wcost_\mathbf{K}(S, C) = \cost(\gamma_S) \leq \cost(\gamma'_S) \leq (1 + \varepsilon) \cost(\gamma_P) = (1 + \varepsilon) \cdot \cost_\mathbf{K}(P, C). \]
    On the other hand, $\gamma_S$ must satisfy some color constraint $K_S \in \mathbf{K}$ and so there exists an allocation $\gamma'_P: P \to C$ satisfying $K_S$ for which $\cost(\gamma_S) \geq (1 - \varepsilon) \cost(\gamma'_P)$. Again, $\cost(\gamma_P) \leq \cost(\gamma'_P)$ and so
    \[ \wcost_\mathbf{K}(S, C) = \cost(\gamma_S) \geq (1 - \varepsilon) \cost(\gamma'_P) \geq (1 - \varepsilon) \cost(\gamma_P) = (1 - \varepsilon) \cdot \cost_\mathbf{K}(P, C)\]
    proves the claim.
\end{proof}

\section{PTAS for constrained \kk-means}\label{appendix:inaba}

We will now see how color coresets may be used to solve the constrained clustering problems indicated above. Recall, that for many constraints an optimal assignment can efficiently be found using min-cost-flows for a given set of centers (see e.g. \cite{DX15}). These constraints include size constraints, must-link constraints, chromatic constraints and $l$-diversity constraints. Arbitrary cannot-link constraints and fairness constraints might be more difficult, but for the latter at least constant factor approximation algorithms exist. So using the color coreset we would like to find a set of centers that admit a cheap assignment. We do this by adjusting a lemma by Inaba et al. \cite{IKI94} to also work with weighted input sets.

\begin{lemma}
    Let $(S, w)$ be a multi-set of weighted points in $\mathbb{R}^d$ and $R = (r_1, \ldots r_m)$ be the results of sampling $m$ points independently from $S$ according to the probability mass function $f(s) = \frac{w(s)}{w(S)}$. Then
    \[ \left\| \mu^w(S) - \mu(R) \right\|^2 < \frac{1}{\delta m w(S)} \sum_{s \in S} w(s) \left\| s - \mu^w(S) \right\|^2 \]
    with probability at least $1 - \delta$ for any $\delta > 0$. Here
    $ \mu^w(S) = \frac{1}{w(S)} \sum_{s \in S} w(s) \, s$
    denotes the weighted mean.
\end{lemma}

\begin{proof}
    The expected outcome of taking a single sample is
    \[ \E(r_i) = \sum_{s \in S} \frac{w(s)}{w(S)} \, s = \frac{1}{w(S)} \sum_{s \in S} w(s) \, s =  \mu^w(S) \]
    with a variance of
    \[ \V(r_i) = \E(\|r_i - \E(r_i)\|^2)
        = \E\left( \| r_i - \mu^w(S) \|^2\right)
        = \sum_{s \in S} \frac{w(s)}{w(S)}\left\| s - \mu^w(S) \right\|^2.
    \]
    By the linearity of the expectation
    $\E(\mu(R)) = \frac{1}{m} \sum_{i = 1}^m \E(r_i) = \mu^w(S)$
    and so
    \begin{align*}
             & \E\left(\| \mu^w(S) - \mu(R) \|^2\right)
        = \E\left(\| \E(\mu(R)) - \mu(R) \|^2\right)
        = \V(\mu(R))                                          \\
        = \; & \V\left(\frac{1}{m} \sum_{i = 1}^m  r_i\right)
        =  \frac{1}{m^2} \sum_{i = 1}^m \V\left(r_i\right)
        = \frac{1}{m w(S)} \sum_{s \in S} w(s)\left\| s - \mu^w(S) \right\|^2.
    \end{align*}
    In other words, the right hand side of the equation in the lemma is $\frac{1}{\delta}$ times the expected value of the left hand side. Thus the claim is a direct result of Markov's inequality.
\end{proof}

This preliminary lemma lets us prove a weighted version of the Inaba lemma.

\begin{lemma}[Weighted Inaba Lemma]
    \label{lem:weighted-inaba}
    Let $(S, w)$ be a multi-set of $n$ weighted points in $\mathbb{R}^d$ and $R = (r_1, \ldots, r_m)$ be the results of sampling $m$ points independently from $S$ according to the probability mass function $f(s) = \frac{w(s)}{w(S)}$. Then
    \[ \sum_{s \in S} w(s) \|s - \mu(R)\|^2 < (1 + \frac{1}{\delta m}) \sum_{s \in S} w(s) \|s - \mu^w(S)\|^2 \]
    with probability at least $1 - \delta$ for any $\delta \in (0, 1]$.
\end{lemma}

\begin{proof}
    We have
    \begin{align*}
        \sum_{s \in S} w(s) \|s - \mu(R)\|^2 = \sum_{s \in S} w(s) \|s - \mu^w(S)\|^2 + w(S) \|\mu(R) - \mu^w(S)\|^2.
    \end{align*}
    By the previous lemma
    \begin{align*}
        w(S) \|\mu(R) - \mu^w(S)\|^2 < \frac{1}{\delta m} \sum_{s \in S} w(s) \left\| s - \mu^w(S) \right\|^2
    \end{align*}
    holds with probability of at least $1 - \delta$. Combining both formulas finishes the proof.
\end{proof}

For large $m$ the centroid of the sampled set $R$ yields a better and better center for the original multiset $S$.

\begin{corollary}
    \label{cor:weighted-inaba}
    Let $(S, w)$ be a multi-set of $n$ weighted points in $\mathbb{R}^d$ and $\varepsilon > 0$. Then for every $m > \frac{1}{\varepsilon}$ there exists a multi-set $R$ of $m$ points from $S$, such that
    \[ \sum_{s \in S} w(s) \|s - \mu(R)\|^2 < (1 + \varepsilon) \sum_{s \in S} w(s) \|s - \mu^w(S)\|^2. \]
\end{corollary}

\begin{proof}
    Apply Lemma \ref{lem:weighted-inaba} for any $\delta$ with $0 < \delta \leq \frac{1}{\varepsilon m} < 1$. Then
    \begin{align*}
        \sum_{s \in S} \| s - \mu(R) \|^2 & < \left(1 + \frac{1}{\delta m}\right) \sum_{s \in S} \| s - \mu^w(S) \|^2                  \\
                                          & < \left(1 + \frac{1}{\frac{1}{\varepsilon m} m}\right) \sum_{s \in S} \| s - \mu^w(S) \|^2 \\
                                          & = \left(1 + \varepsilon\right) \sum_{s \in S} \| s - \mu^w(S) \|^2
    \end{align*}
    holds with probability $1 - \delta > 0$. Since this probability is positive there must be at least one such multi-set $R$ and the claim follows.
\end{proof}

The important part is that this let's us find a good center without involving the weights of the original set. We can use this to find a good center for every cluster of an optimal clustering by fully enumerating partitions of small (weighted) subsets. Let $\mathbf{K}$ be an color constraint and $S$ a color coreset with weights $w: S \to \mathbb{N}$ and colors $f: S \to \{1, \ldots, l\}$. A $k$-clustering of $S$ corresponds to a collection of weighted sets $(S_1, w_1), \ldots, (S_k, w_k)$ such that $S_i \subset S$ for all $i$ and $\sum_{i = 1}^k w_i(s) = w(s)$ for all $s \in S$ (for ease of notation just assume that $w_i(s) = 0$ if $s \not\in S_i$). Assume that $(S_1, w_1), \ldots, (S_k, w_k)$ has optimal cost $OPT$ among all clusterings satisfying $\mathbf{K}$.
By Corollary \ref{cor:weighted-inaba} there exist $k$ multisets $R_1, \ldots, R_k \subset S$ of size $\lceil \frac{2}{\varepsilon} \rceil$ such that
\[ \sum_{s \in S_i} w_i(s) \|s - \mu(R_i)\|^2 < (1 + \varepsilon) \sum_{s \in S_i} w_i(s) \|s - \mu^w(S_i)\|^2 \]
for all $i$. Summing these up,
\[ \sum_i \sum_{s \in S_i} w_i(s) \|s - \mu(R_i)\|^2 < (1 + \varepsilon) \sum_i \sum_{s \in S_i} w_i(s) \|s - \mu^w(S_i)\|^2 = (1 + \varepsilon) \cdot OPT \]
shows that an optimal assignment of $S$ to the centers $\mu(R_1), \ldots, \mu(R_k)$ is a $(1 + \varepsilon)$-approximation. Finding these subsets $R_1, \ldots, R_k$ can be done by a full enumeration over all equal size $k$-partitionings of all multisets consisting of $k \lceil \frac{2}{\varepsilon}\rceil$ elements in $S$. After passing to the centroids this yields a list of $|S|^{O(\frac{k}{\varepsilon})}$ sets of $k$ centers each, of which at least one will be good. Now one can just compute an optimal assignment to each of those center sets and choose the best one. As mentioned in the beginning, many constrained problems admit efficient optimal assignment algorithms, so this is not an issue.

\ptasInaba*

\begin{proof}
    Just compute a $(k, \ell, \frac{\varepsilon}{3})$-color coreset in a stream (which is possible by Corollary~\ref{thm:size-color-coresets}) and then apply the above algorithm for $\varepsilon' = \frac{\varepsilon}{3}$. Assigning points to centers is relatively cheap and so the running time to get an assigment with cost at most $(1 + \frac{\varepsilon}{3})^2 \leq (1 + \varepsilon)$ times that of the global optimum is $(\ell k \varepsilon^{-d} \log n)^{O(\frac{k}{\varepsilon})}$ if we use the constructions from Theorem \ref{thm:size-color-coresets}.
\end{proof}

\end{document}